\documentclass[11pt]{article}
\usepackage{amsfonts,amsthm,amssymb,latexsym,amsmath,epsfig,color}
\usepackage{verbatim,graphicx}
\usepackage{dsfont}
\usepackage{a4wide}
\usepackage[utf8]{inputenc}

\theoremstyle{plain}
\newtheorem{theorem}{Theorem}
\newtheorem{corollary}{Corollary}
\newtheorem{lemma}{Lemma}

\newtheorem{fact}{Fact}

\begin{document}

\title{Mechanism design for aggregating energy consumption and quality of service in speed scaling scheduling\thanks{A preliminary version of this article appeared in \emph{Proc.\ of the 9th Conf.\ on Web and Internet Economics} (WINE 2013).}}

\author{Christoph D\"urr\thanks{Christoph.Durr@LIP6.fr, Sorbonne Universités, UPMC Univ Paris 06, CNRS, LIP6, Paris, France}
\and
{\L}ukasz Je{\.z}\thanks{lje@cs.uni.wroc.pl,
Eindhoven Institute of Technology (NL) and Institute of Computer Science, University of  Wroc{\l}aw (PL).}
\and
\'Oscar C. V\'asquez\thanks{oscar.vasquez@usach.cl,
University of Santiago of Chile, Department of Industrial Engineering.}}

\date{}

\maketitle

\begin{abstract}
  We consider a strategic game, where players submit jobs to a~machine
  that  executes  all  jobs  in  a~way  that  minimizes  energy  while
  respecting  the given deadlines.   The energy  consumption is  then
  charged to the  players in some way.  Each  player wants to minimize
  the sum  of that  charge and of  their job's deadline  multiplied by
  a~priority   weight.   Two   charging  schemes   are   studied,  the
  \emph{proportional cost share} which does not always admit pure Nash
  equilibria, and  the \emph{marginal  cost share}, which  does always
  admit  pure Nash  equilibria,  at  the price  of  overcharging by  a
  constant factor.
\end{abstract}

\paragraph{keywords}
scheduling, energy management, quality of service, optimization, mechanism design.

\section{Introduction}
In   many  computing  systems,  maximizing quality of service comes generally at the price of a high energy consumption. This is also the case
for the speed  scaling scheduling model considered in  this paper.  It
has  been introduced  in~\cite{YaoDemmersShenker1995}, and  triggered a
lot     of    work     on    offline     and     online    algorithms;
see~\cite{albers2010energy} for an overview.

The online  and offline optimization problem for  minimizing flow time
while respecting a maximum energy consumption has been studied for the
single                machine                setting                in
\cite{pruhs2008getting,albers2007energy,bansal2007speed,chan2010non}
and for  the parallel machines setting  in \cite{angel2012energy}. For
the  variant  where  an  aggregation   of  energy  and  flow  time  is
considered, polynomial approximation algorithms have been presented in
\cite{carrasco11:_energ_aware_sched_weigh_compl,bansal2012improved,megow13:dual}.

In  this paper  we  propose to  study  this problem  from a  different
perspective, namely  as a strategic  game. In society  many ecological
problems are  either addressed in a centralized manner,  like forcing
citizens to sort household waste, or in a  decentralized manner, like
tax  incentives to  enforce ecological  behavior. This  paper proposes
incentives for a  scheduling game, in form of an energy cost charging
scheme.

Consider a scheduling problem for a single processor, that can run at
variable speed, such as the modern microprocessors Intel SpeedStep, AMD PowerNow! or IBM EnergyScale.
Higher speed means  that jobs finish earlier  at the
price  of a  higher energy  consumption.
Each  job has  some  workload,  representing  a number  of
instructions to execute, and a  release time before which it cannot be
scheduled.   Every  user  submits  a~single  job  to  a  common
processor, declaring  the job parameters, together with  a deadline,
that the player chooses at his convenience.

The processor  will schedule the submitted jobs  preemptively, so that
all release times  and deadlines are respected and  the overall energy
usage is minimized.  The energy consumed by the schedule needs to be
charged to the users. The individual goal of each user is to minimize
the  sum  of the  energy  cost share  and  of  the requested weighted deadline. The weight is a private priority factor representing the individual importance of a small deadline.
This factor includes implicitly a~conversion factor
that allows for an aggregation of the deadline and energy consumption into a single individual penalty.

In a  companion paper~\cite{durr2014penaltyscheduling} we study this game from
the point of view of the game regulator, in a different setting. The players announce with their job their priority factors, and the regulator gets to decide on the completion time of the jobs.  The usual questions one asks for such a game, is the existence of a cost sharing mechanism that would be truthful on the priority factors and which charge to the players amounts that sum up to a value comparable within a constant factor to the actual energy consumption of the schedule.  This contrasts with the setting considered in this paper, where the player's strategies are the job's deadlines.

\section{The model}
Formally, we  consider a non-cooperative  game with $n$ players  and a
regulator.   The regulator  manages  the machine  where  the jobs  are
executed.  Each player has a job  $i$ with a workload $w_i$, a release
time $r_i$  and a  priority $p_i$, representing  a quality  of service
coefficient.  The player  submits  its job  together  with a  deadline
$d_i>r_i$ to the regulator.  Workloads, release times and deadlines are public
information   known  to   all  players,   while  quality   of  service
coefficients can be private.

The regulator implements some  \emph{cost sharing mechanism}, which is
known to all users. This mechanism defines a cost share function $b_i$
specifying how much  player $i$ is charged. The penalty of player
$i$ is the sum of two values:
his \emph{energy cost share} $b_i(w,r,d)$ defined by the mechanism,
where $w=(w_1,\ldots,w_n), r=(r_1,\ldots,r_n), d=(d_1,\ldots,d_n)$ are the input values,
and his \emph{waiting cost}, which can be either $p_i d_i$ or $p_i (d_i-r_i)$;
we use the former waiting cost throughout the article but all our results apply to both settings.
The sum of all player's penalties, i.e., energy cost shares and waiting costs
will be called the \emph{utilitarian social cost}.

The regulator computes a minimum  energy schedule for a single machine in the speed scaling  model. In this model at any point in time $t$ the processor can run at some speed $s(t) \geq 0$. As a result, for any time interval  $I$,  the  workload  executed in~$I$ is  $\int_{t\in I}  s(t) \, \mathrm{d}t$ at the price of an energy consumption  valuated at $\int_{t\in I}  s(t)^\alpha\,\mathrm{d}t$
for some fixed physical constant $\alpha \in [2,3]$ which is device dependent~\cite{brooks2000power}.

The sum of the energy used by this optimum schedule and of all the players'
waiting costs will be called the \emph{effective social cost}.

The  minimum energy  schedule  can  be computed  in  time $O(n^2  \log
n)$~\cite{LiOn3200discrete}  and  has  (among  others)  the  following
properties~\cite{YaoDemmersShenker1995}.  The jobs in the schedule are
executed by  preemptive earliest deadline  first order (EDF),  and the
speed  $s(t)$  at  which  they  are  processed  is  piecewise  constant.
Preemptive EDF means that at every time point among all jobs which are
already  released and  not yet  completed, the  job with  the smallest
deadline is executed, using job indexes to break ties.

The cost  sharing mechanism defines the game  completely.  Ideally, we
would  like  the  game  and   the  mechanism  to  have  the  following
properties.
\begin{description}
\item[existence of pure  Nash equilibria] This  means that  there is always a
  strategy  profile vector $d$  such that  no player  can unilaterally
  deviate  from  his  strategy  $d_i$ while  strictly  decreasing  his
  penalty.
\item[budget balance] The mechanism is $c$-budged balanced,
when the sum of the cost shares is no smaller than the total energy consumption
and no larger than $c$ times the energy consumption. Ideally we would like $c$ to be close to $1$.
\end{description}

In  the sequel  we  introduce  and study  two  different cost  sharing
mechanisms,  namely  \textsc{Proportional  Cost Sharing}  where  every
player pays  exactly the  cost generated during  the execution  of his
job, and  \textsc{Marginal Cost Sharing}  where every player  pays the
increase of energy cost generated by adding this player to the game.

\section{Proportional cost sharing}

The  proportional cost sharing  is the  simplest budget  balanced cost
sharing scheme one  can think of. Every player  $i$ is charged exactly
the  energy consumed during  the execution  of his  job. Unfortunately
this mechanism does not behave well as we show in
Theorem~\ref{thm:prop}.

\begin{fact}
  In a single player game, the player's penalty is minimized by the deadline
\[
        r_1 + w_1(\alpha-1)^{1/\alpha} p_1^{-1/\alpha}.
\]
\end{fact}
\begin{proof}
  If player  $1$ chooses deadline $d_1=r_1  + x$, then his job is processed in the time interval $[r_1,\ r_1 + x]$ at speed $w_1/x$. Therefore
  his penalty is
\[
         p_1 (r_1 + x)  + x^{1-\alpha} w_1^\alpha.
\]
Differentiating this expression in $x$, and using the  fact that the penalty
is  concave in  $x$ for  any $x>0$  and $\alpha>0$,  we have  that the
optimal  $x$ for  the player will set the derivative to zero. This implies the claimed deadline.
\end{proof}

If there are at least two players however, the game does not have nice
properties as we show now.
\begin{theorem} \label{thm:prop}
  The \textsc{Proportional Cost Sharing}  does not always admit a pure
  Nash equilibrium.
\end{theorem}


The  proof consists of an example consisting of two identical
players  with identical  jobs, say  $w_1=w_2=1$, $r_1=r_2=0$  and
$p_1=p_2=1$.  First we determine the best  response of player  1 as a
function of  player 2. Then we  conclude that there  is no  pure Nash
equilibrium.


\begin{lemma}\label{lem:p1-resp-and-min}
Given the second player's choice $d_2$, the penalty of the first player
as a~function of his choice $d_1$ is given by
\begin{equation}\label{eq:p1-penalty}
  f(d_1) = \begin{cases}
  f_1(d_1) = d_1+(d_1-d_2)^{1-\alpha}   &\mbox{if } d_1 \geq2d_2\\
  f_2(d_1) = d_1+(\frac{d_1}{2})^{1-\alpha} &\mbox{if } d_2 \leq d_1 \leq 2d_2 \\
  f_3(d_1) = d_1+(\frac{d_2}{2})^{1-\alpha}	&\mbox{if } \frac{d_2}{2} \leq d_1 \leq d_2 \\
  f_4(d_1) = d_1+d_1^{1-\alpha}     &\mbox{if } d_1 \leq \frac{d_2}{2}.
 \end{cases}
\end{equation}
The local minima of $f(d_1)$ are summarized in Table~\ref{tbl:minima},
and the penalties corresponding to player 1 picking these minima are
illustrated in Figure~\ref{fig:bestResponse}.
\begin{table}[t]
  \begin{tabular}[c]{lllll}
    argument							& value									& applicable range \\ \cline{1-3}
    $d_1^1 = d_2+\left(\alpha-1\right)^{1/\alpha}$    & $g_1(d_2) = d_2+\alpha \left(\alpha-1\right)^{1/\alpha-1}$     & $d_2 \leq \tau_3$ & $ \tau_3=\left(\alpha-1\right)^{1/\alpha}$ \\
    $d_1^2 = 2\left(\frac{\alpha-1}{2}\right)^{1/\alpha}$ & $g_2(d_2) = \alpha\left(\frac{\alpha-1}{2}\right)^{1/\alpha-1}$  & $\tau_1 \leq d_2 \leq \tau_5$ & $ \tau_1 = \left(\frac{\alpha-1}{2}\right)^{1/\alpha},\:\tau_5 = 2 \tau_1$ \\
    $d_1^3 = \frac{d_2}{2}$				& $g_3(d_2) = d_2/2+(d_2/2)^{1-\alpha}$				 & $d_2\leq \tau_6$ \\
    $d_1^4 = (\alpha-1)^{1/\alpha}$     & $g_4(d_2) = \alpha(\alpha-1)^{1/\alpha-1}$         & $d_2 \geq \tau_6$
    & $\tau_6 = 2(\alpha-1)^{1/\alpha}$ \\
  \end{tabular}
  \caption{The local minimum in the range of $f$ corresponding to $f_i$ is a~function of $\alpha$ and $d_2$, which we denote
  by $d_1^{i}$.  The value at such local minimum is again a~function of $\alpha$ and $d_2$, which we denote by $g_i(d_2)$.
  These are only conditional minima: they exist if and only if the condition given in the last column is satisfied.}
  \label{tbl:minima}
\end{table}
\end{lemma}

\begin{figure}[ht]
  \centering
  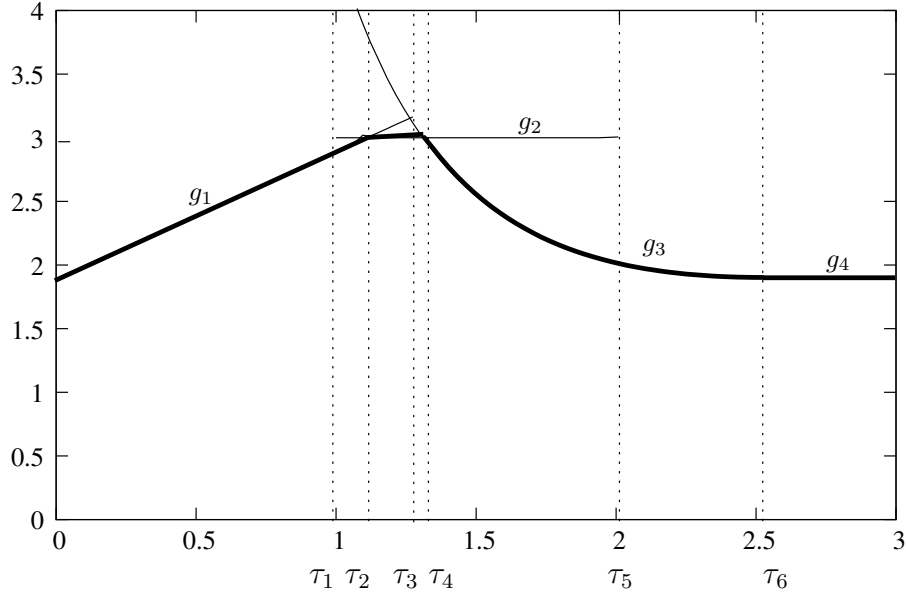
  \caption{First player's penalty (in bold) when choosing his best response as a function of second player's strategy $d_2$, here for $\alpha=3$.
  	As seen on the plot, $\tau_2$ ($\tau_4$), formally defined in Lemma~\ref{lem:d2*}, is the argument for which $g_1$ and $g_2$
  	($g_2$ and $g_3$) attain the same value.}
  \label{fig:bestResponse}
\end{figure}

\begin{proof}
 Formula~\eqref{eq:p1-penalty} follows by a~straightforward case inspection.
 Then, to find all the local minima of $f$, we first look at the behavior of each of $f_i$,
 finding its local minima in their respective intervals, and afterwards we inspect the
 border points of these intervals.

\begin{description}
\item[Range of $f_1$:] The derivative of $f_1$ is
\begin{equation}\label{eq:f6-derivative}
   f_1'(d_1) = 1-(\alpha-1) (d_1-d_2)^{-\alpha},
\end{equation}
whose derivative in turn is positive for $\alpha>1$.
Hence, $f_1$ has a~local minimum at $d_1^1$ as specified.
The existence of this local minimum requires $d_1^2 \geq 2d_2$.
\item[Range of $f_2$:] The derivative of $f_2$ is
\begin{equation}\label{eq:f4-derivative}
   f_2'(d_1) = 1-\frac{\alpha-1}{2} (d_1/2)^{-\alpha},
\end{equation}
whose derivative in turn is positive for $\alpha>1$.
Hence, $f_2$ has a~local minimum at $d_1^2$ as specified.
The existence of this local minimum requires $d_2 \leq d_1^2 \leq 2d_2$,
which is equivalent to $d_1^3/2 \leq d_2 \leq d_1^2$.
\item[Range of $f_3$:] $f_3$ is an increasing function,
and therefore it attains a~minimum value only at the lower end of its range, $d_1^2$.
However, if $d_1^3$ is to be a~local minimum of $f$, there can be no local minimum
of $f$ in the range of $f_4$ (immediately to the left), so the applicable range of
$d_1^3$ is the complement of that of $d_1^4$.
\item[Range of $f_4$:] The derivative of $f_4$ is
\[
    f_4'(d_1) = 1-(\alpha-1)d_1^{-\alpha},
\]
whose derivative in turn is positive for $\alpha>1$.
Therefore, $f_4$ has a~local minimum at $d_1^4$ as specified.
Since we require that this local minimum is within the range
where $f$ coincides with $f_4$, the necessary and sufficient
condition is $d_1^4 \leq d_2/2$.
\end{description}

Now let us consider the boundaries of the ranges of each $f_i$.
Since $f_3$ is strictly increasing, the
border point of the ranges of $f_3$ and $f_2$ is not a~local minimum
of~$f$.  This leaves only the border point $d_1^3=2d_2$ of the ranges of
$f_2$ and $f_1$ to consider.  Clearly, $d_1^3$ is a~local minimum of~$f$
if and only if $f_2'(d_1^3) \leq 0$ and $f_1'(d_1^3) \geq 0$.
However, by~\eqref{eq:f4-derivative}, $f_2'(d_1^3)=2-(\alpha-1)d_2^{-\alpha}$,
and by~\eqref{eq:f6-derivative}, $f_1'(d_1^3)=2-2(\alpha-1)d_2^{-\alpha}<f_2'(d_1^3)$,
so $d_1^3$ is not a~local minimum of~$f$ either.
\end{proof}


Note that the range of $g_4$ is disjoint with the ranges of $g_2$ and $g_1$,
and with the exception of the shared border value $2(\alpha-1)^{1/\alpha}$,
also with the range of $g_3$.  However, the ranges of $g_3$, $g_2$ and $g_1$
are not disjoint.  Therefore, we now focus on their shared range, and determine
which of the functions gives rise to the true local minimum.

\begin{lemma}\label{lem:d2*}
The functions $g_2(d_2)$ and $g_4(d_2)$ are constant, the function $g_1(d_2)$ is
increasing and linear, and the function $g_3(d_2)$ is
decreasing for $d_2\leq \tau_6$.
Moreover, there exist two unique values $\tau_2$ and $\tau_4$ with
\[
g_1(\tau_2) = g_2(\tau_2)
 \mbox{ and } g_2(\tau_4) = g_3(\tau_4).
\]
In addition we have
\[
\tau_1 < \tau_2 < \tau_3 < \tau_5 < \tau_6
\]
and
\[
  \tau_2 \leq \tau_4 < \tau_5.
\]
\end{lemma}

\begin{proof}
  It follows from their definitions in Table~\ref{tbl:minima} that
  $g_2$ and $g_4$ are constant and $g_1$ strictly increasing.
In order to show that the $g_3(d_2)$ is decreasing in the range $0\leq d_2 \leq \tau_6$, we show that its derivative in $d_2$ is non-positive, namely
\begin{align*}
      \frac12 - 2^{\alpha-1}(\alpha-1) d_2^{-\alpha}
      &\leq 0
                    & \equiv \\
      1
      &\leq
      2^{\alpha}(\alpha-1) d_2^{-\alpha}
                    & \equiv \\
      d_2^{\alpha}
      &\leq
      2^{\alpha}(\alpha-1)
                    & \equiv \\
      d_2
      &\leq
      2(\alpha-1)^{1/\alpha} = \tau_6.
\end{align*}

We define $\tau_2$ as the unique root of $g_1(d_2)=g_2(d_2)$,
  namely
  \[
  \tau_2 = \alpha(\alpha-1)^{1/\alpha-1}(2^{1-1/\alpha}-1).
  \]

Now we show that there is value $\tau_4$ such that $g_3(\tau_4)=g_2$.
This follows from the fact that $g_3$ is continuous and decreasing, that its limit at $d_2 \rightarrow 0$ is $\infty$ and that $g_3(\tau_6) = g_4 < g_2$.

For the bounds on $\tau_4$ we need to show
\[
      g_3(\tau_5) < g_2 \leq g_3(\tau_2).
\]
We start with the left inequality:
\begin{align*}
g_3\left(  \tau_5 \right)
&=  g_3\left(2\left(\frac{\alpha-1}2 \right)^{1/\alpha} \right)
    \\&=  \left(\frac{\alpha-1}{2} \right)^{1/\alpha} \left(1+\frac{2}{\alpha-1} \right) \\
&=  \left(\frac{\alpha-1}{2} \right)^{1/\alpha-1} \frac{\alpha+1}2
    \\&<  \left(\frac{\alpha-1}{2} \right)^{1/\alpha-1} \alpha.
\end{align*}

For the right inequality, we first note that
\begin{align*}
  g_3\left(\tau_2\right) &= \tau_2/2+(\tau_2/2)^{1-\alpha}
  \\&= \tau_2 (1/2+(\tau_2)^{-\alpha}2^{\alpha-1}) \\
    &= \alpha(\alpha-1)^{1/\alpha-1}(2^{1-1/\alpha}-1)\left(\frac{2^{\alpha-1}}{(\alpha(\alpha-1)^{1/\alpha-1}(2^{1-1/\alpha}-1))^\alpha}+\frac{1}{2}    \right),
\end{align*}
hence $ g_3\left(\tau_2\right) \geq \alpha \left(\frac{\alpha-1}{2} \right)^{1/\alpha-1}$
is equivalent to
\begin{align*}
(2^{1-1/\alpha}-1)\left(\frac{2^{\alpha-1}}{\alpha^\alpha(\alpha-1)^{1-\alpha}(2^{1-1/\alpha}-1)^\alpha}+\frac{1}{2}
\right)&\geq2^{1-1/\alpha} & \equiv \\
(2^{1-1/\alpha}-1)\left(\frac{2^{\alpha-1}}{\alpha \alpha^{\alpha-1}(\alpha-1)^{1-\alpha}(2^{1-1/\alpha}-1)^\alpha}-\frac{1}{2}    \right)&\geq1& \equiv \\
\frac{2^{\alpha-1}(\alpha-1)^{\alpha-1} \alpha^{1-\alpha}}{\alpha(2^{1-1/\alpha}-1)^\alpha}-\frac{1}{2}  &\geq\frac{1}{2^{1-1/\alpha}-1} & \equiv \\
 \frac{1}{\alpha}\frac{\left(2-\frac{2}{\alpha}\right)^{\alpha-1}}{(2^{1-1/\alpha}-1)^\alpha}&\geq\frac{1}{2^{1-1/\alpha}-1}+\frac{1}{2}& \equiv \\
 \frac{1}{\alpha}\frac{\left(2-\frac{2}{\alpha}\right)^{\alpha-1}}{\left(\frac{2-2^{1/\alpha}}{2^{1/\alpha}}\right)^\alpha}&\geq\frac{2^{1-1/\alpha}+1}{2(2^{1-1/\alpha}-1)}& \equiv \\
 \frac{1}{\alpha}\frac{\left(2-\frac{2}{\alpha}\right)^{\alpha-1}}{(2-2^{1/\alpha})^{\alpha-1}(2-2^{1/\alpha})} 2&\geq\frac{(2+2^{1/\alpha})2^{-1/\alpha}}{2(2-2^{1/\alpha})2^{-1/\alpha}}& \equiv \\
 \frac{1}{\alpha}\left(\frac{2-\frac{2}{\alpha}}{2-2^{1/\alpha}}\right)^{\alpha-1}&\geq\frac{2+2^{1/\alpha}}{4}.
\end{align*}

We claim that for $\alpha\geq2$ we have
\begin{equation}
\label{eq:relationfunction}
\frac{2-\frac{2} {\alpha} }{2-2^{1/\alpha}} \geq  \frac1{2-\sqrt2}.
\end{equation}
Since we  have equality at  $\alpha=2$, it suffices to prove that
the left hand side is increasing.  To this end, we consider its derivative
\[
 \frac{2\left(2-2^\frac{1}{a} - \ln 2 \cdot (1-\frac{1}{\alpha}) \cdot 2^\frac{1}{\alpha}\right)} {\left(\alpha(2-2^\frac{1}{\alpha})\right)^2}
 =
 \frac{4 - 2^{1+\frac{1}{\alpha}} \cdot \left(1+(1-\frac{1}{\alpha})\ln 2\right)}
 {\left(\alpha(2-2^\frac{1}{\alpha})\right)^2},
\]
and note that its enumerator is increasing in  $\alpha$ and equals $0$ for $\alpha=1$.  Thus~\eqref{eq:relationfunction} holds.
%
%

This permits us to define
\[
    z(\alpha):=(2-\sqrt2)^{1-\alpha}/\alpha,
\]
and to upper bound
\[
    \frac{1}{\alpha}\left(\frac{2-\frac{2}{\alpha}}{2-2^{1/\alpha}}\right)^{\alpha-1}
\geq
    z(\alpha).
\]
In order to lower bound
\[
 z(\alpha)\geq \frac{2+2^{1/\alpha}}{4}
\]
for  $\alpha\geq2$, it  suffices to  show that  $z$ is  increasing with
$\alpha$, since  the right hand  side is decreasing with  $\alpha$.
Its derivative is
\[
z'(\alpha)=-\frac{(2-\sqrt2)^{1-\alpha} ( 1+\alpha \ln(2-\sqrt2))}{\alpha^2}.
\]
Observe that $\alpha\ln(2-\sqrt2)<-1$
for every $\alpha\geq 2$, and therefore $z'$ is positive and $z$ is
increasing as required.
This concludes the existence of $\tau_4$ with the required properties.

It remains to prove the remaining relations among $\tau$'s.
We begin with
\begin{align*}
\tau_1 & < \tau_2
          & \equiv \\
\left(\frac{\alpha-1}{2} \right)^{1/\alpha}
&<
\alpha(\alpha-1)^{1/\alpha-1}(2^{1-1/\alpha}-1)
          & \equiv \\
2^{-1/\alpha}
&<
\frac{\alpha}{\alpha-1}(2^{1-1/\alpha}-1)
          & \equiv \\
1
&<
\frac{\alpha}{\alpha-1}(2-2^{1/\alpha})
      & \equiv \\
\frac{\alpha-1}\alpha
&<
2-2^{1/\alpha}
      & \equiv \\
2^{1/\alpha} - 1/\alpha
&<
1.
\end{align*}
To prove this inequality, we note that it holds as an equality in the limit $\alpha \rightarrow \infty$, and that the left hand side is increasing in $\alpha$, since its derivative is
\[
\frac{1-2^{1/\alpha} \ln(2)}{\alpha^2},
\]
which is positive for $\alpha\geq 2$.

Now we show
\begin{align*}
\tau_2
&<
\tau_3
      & \equiv \\
\alpha(\alpha-1)^{1/\alpha-1}(2^{1-1/\alpha}-1) & < (\alpha-1)^{1/\alpha} & \equiv \\
\alpha(2^{1-1/\alpha}-1) & < \alpha - 1
        & \equiv \\
2^{1-1/\alpha}-1 & < 1-1/ \alpha
        & \equiv \\
2^{1-1/\alpha} + 1/\alpha & < 2.
\end{align*}
We observe that the left hand side has value 2 both at $\alpha=1$ and in the limit $\alpha \rightarrow \infty$. To conclude that this holds for all
$\alpha \in [1,\ \infty)$, we inspect the derivative of the left hand side with respect to $\alpha$, which is
\[
\frac{2^{1-1/\alpha} \ln(2)  - 1}{\alpha^2}.
\]
Since $2^{1/\alpha-1}$ is monotone, there is a unique value $\alpha_0 = \frac{\ln 2}{\ln2 + \ln\ln 2} \approx 2.1221 $ such that
\[
    2^{1/\alpha_0-1} = \ln(2).
\]
In conclusion, the derivative is negative for $1 \leq \alpha < \alpha_0$ and then positive for $\alpha > \alpha_0$, so inside the interval the function never exceeds $2$.

The inequality $\tau_3 < \tau_5$ follows from the equality $\tau_5 = 2^{1-1/\alpha} \tau_3$, while the inequality $\tau_5 < \tau_6$ follows from the equality $\tau_6 = 2^{1/\alpha} \tau_5$.
This concludes the proof of the lemma.
\end{proof}

With Lemma~\ref{lem:p1-resp-and-min} and Lemma~\ref{lem:d2*},
whose statements are summarized in Table~\ref{tbl:minima} and Figure~\ref{fig:bestResponse},
we can finally determine what is the best response of the first player
as a~function of $d_2$.  The following corollary follows from the definitions of $\tau_2,\tau_4$ and the inequalities on $\tau_1,\ldots,\tau_6$.

  \begin{corollary} \label{cor:bestResponse}
    The best response for player 1 as function of $d_2$ is
\begin{alignat*}{3}
d_1^1 &= d_2+(\alpha-1)^{1/\alpha} &&\text{\quad if \quad} & 0 &< d_2  \leq \tau_2,\\
d_1^2 &= 2\left(\frac{\alpha-1}{2} \right)^{1/\alpha}  &&\text{\quad if \quad} &  \tau_2 &< d_2 \leq {\tau_4},\\
d_1^3 &= \frac{d_2}{2} &&\text{\quad if \quad} & {\tau_4} &<d_2 \leq \tau_6,\\
d_1^4 &= (\alpha-1)^{1/\alpha} &&\text{\quad if \quad} &  \tau_6&<d_2.\\
\end{alignat*}
 \end{corollary}

\begin{figure}[t]
  \centering
  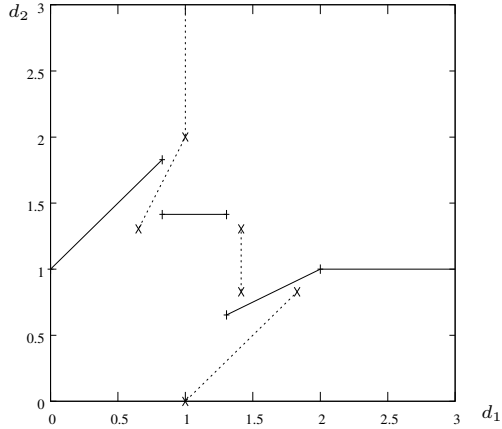
  \caption{Best response  of player 1  as function of $d_2$ (dashed  lines), and best response of player 2 as function of $d_1$ (solid lines). Here for $\alpha=3$.}
  \label{fig:bestResponse12}
\end{figure}

By the symmetry  of the players, the second  player's best response is
in  fact  an  identical  function  of  $d_1$  as  the  one  stated  in
Corollary~\ref{cor:bestResponse}. By straightforward inspection it follows
that there is no fix point $(d_1,d_2)$ to this game, which concludes the proof of Theorem \ref{thm:prop}, see Figure~\ref{fig:bestResponse12} for illustration.


\section{Marginal cost sharing}
In  this section  we propose  a  different cost  sharing scheme,  that
improves on  the previous one  in the sense  that it admits  pure Nash
equilibria,  but does so at the  price  of  overcharging  by a
constant factor.

Before  we  give the  formal  definition  we  need to  introduce  some
notations. Let $\mathrm{OPT}(d)$ be the energy minimizing schedule for
the   given  instance,  and   $\mathrm{OPT}(d_{-i})$  be   the  energy
minimizing  schedule for  the instance  where job  $i$ is  removed. We
denote by $E(S)$ the energy cost of schedule $S$.

In the  marginal cost sharing  scheme, player $i$  pays the penalty function
\[p_id_i+E(\mathrm{OPT}(d))  -  E(\mathrm{OPT}(d_{-i})).\]  This  scheme
defines     an     exact     potential    game     by     construction
\cite{Monderer.Shapley:96:potential_games}.  Formally, let $n$  be the
number of players, $D=\{d | \forall j: d_j > r_j \}$ be the set of action profiles
(deadlines) over the action sets $D_{i}$ of each player.


Let us denote the effective social cost corresponding to a strategy profile $d$
by $\Phi(d)$.  Then
\[\Phi(d)= \sum_{i=1}^n p_id_i+E(\mathrm{OPT}(d)).\]

Clearly, if a player $i$ changes its strategy $d_i$ and his penalty decreases by some amount $\Delta$,
then the effective social cost decreases by the same amount $\Delta$, because $E(\mathrm{OPT}(d_{-i}))$ remains unchanged.
%
%

\subsection{Existence of Equilibria}
While the best response function is not continuous in the strategy profile,
precluding the use of Brouwer's fixed-point theorem, existence of pure Nash
equilibria can nevertheless be easily established.

To this end, note that the global minimum of the effective social cost, if it exists,
is a pure Nash equilibrium.  Its existence follows from (1) compactness
of a non-empty sub-space of strategies with bounded social cost and
(2) continuity of~$\Phi$.

For (2), note that $\sum_i p_i d_i$ is clearly continuous in $d$,
and hence $\Phi(d)$ is continuous if $E(\mathrm{OPT}(d))$ is.
The continuity of the latter follows from the fact that $E(\mathrm{OPT}(d))$ is the solution to a linear program with a convex objective function, with an optimum being continuous in $d$ \cite{bansal2007speed}.

For (1), let $d'$ be any (feasible) strategy profile such that $d'_i>r_i$
for each player $i$.  The subspace of strategy profiles $d$ such that
$\Phi(d) \leq \Phi(d')$ is clearly closed, and bounded due to the $p_i d_i$
terms.  Thus it is a compact subspace that contains the global minimum of $\Phi$.

\subsection{Convergence can take forever}

In this game the strategy set is infinite.  Moreover, the convergence time
can be infinite as we demonstrate below in
Theorem~\ref{thm:convergenceUnbounded}.
Note that this also proves that in general there are no dominant strategies in this game.

\begin{theorem} \label{thm:convergenceUnbounded}
  For  the  game  with   the  marginal  cost  sharing  mechanism,  the
  convergence time to reach a pure Nash equilibrium can be unbounded.
\end{theorem}

\begin{proof}
The  proof is  by  exhibiting again  the  same small  example, with  2
players, release times $0$, unit weights, unit penalty factors, and $\alpha>2$.
From the previous section, we know that the game admits a pure Nash equilibrium, and by symmetry of the players, in fact two pure Nash equilibria. Following \cite{durr2014penaltyscheduling}, the  first one is
\[
   d_1 = \left( \frac{\alpha-1}{2} \right) ^{1/\alpha}, \: d_2 = d_1 +
(\alpha-1)^{1/\alpha},
\]
while the second one is symmetric for players $1$ and $2$.

In  the remainder  of the  proof,  we assume  that player  1 chooses  a
deadline  which is  close  to  the pure  Nash  equilibrium above.   By
analyzing the best responses of  the players, we conclude that after a
best response  of player 2, and then  of player 1 again,  he chooses a
deadline  which is  even closer  to the  pure Nash  equilibrium above
but still different from it,
leading  to an infinite  convergence sequence  of best  responses.

Now suppose that $d_1   =  \delta   \left(   \frac{\alpha-1}{2}  \right)
^{1/\alpha}$ for some $1<\delta< 2^{1/\alpha}$. What is the best response for player 2?

\begin{lemma} \label{lem:choice2}
Given the first player's choice $d_1$, the penalty of the second player
as a~function of his choice $d_2$ is given by
\begin{equation*}
  h(d_2,d_1) = \begin{cases}
  h_1(d_2,d_1) = d_2 +d_2^{1-\alpha}+(d_1-d_2)^{1-\alpha}-d_1^{1-\alpha} &\mbox{if } d_2 \leq \frac{d_1}{2} \\
  h_2(d_2,d_1) = d_2 +(2^\alpha-1) d_1^{1-\alpha}  &\mbox{if } \frac{d_1}{2} \leq d_2 \leq d_1 \\
  h_3(d_2,d_1) = d_2 +2^\alpha d_2^{1-\alpha}-d_1^{1-\alpha}&\mbox{if } d_1 \leq d_2 \leq 2d_1 \\
  h_4(d_2,d_1) = d_2 +(d_2-d_1)^{1-\alpha}		&\mbox{if } d_2 \geq 2d_1, \\
 \end{cases}
\end{equation*}
     and the best response for player 2 as function of $d_1$ is
     \begin{equation}
     \label{eq:brpla2}
     d_1+(\alpha-1)^{1/\alpha}=(\alpha-1)^{1/\alpha}(1+2^{-1/\alpha}\delta)
     \end{equation}
\end{lemma}

\begin{proof}
We first  analyze the behavior of  each of $h_i$,  finding their local
minima  in the respective  intervals, and  afterwards  we show that the equation \eqref{eq:brpla2} defines the best response for player 2 as function of $d_1$. For convenience we omit parameter $d_1$ in each function $h_i$.
Figure \ref{fig:convergence} illustrate the best response for player $2$ when player $1$ chooses $d_1=(\frac{\alpha-1}2)^{1/\alpha}\delta$ for some $\delta>1$.

\begin{description}
\item[Range of $h_4$:]
The derivative of $h_4$ in $d_2$ is
\[
 1-(\alpha-1)(d_2-d_1)^{-\alpha},
\]
which is zero exactly for the value
\[
d_1 + (\alpha-1)^{1/\alpha}.
\]
As the second derivative $h''_1(d_2) = \alpha(\alpha-1)(d_2-d_1)^{-\alpha-1}$ is
positive, the choice $d_2=d_1 + (\alpha-1)^{1/\alpha}$ minimizes the penalty among $d_2\geq
2d_1$. Therefore, $h_1$ has a~local minimum if
\[(\alpha-1)^{1/\alpha} \geq d_1=\delta \left( \frac{\alpha-1}{2} \right) ^{1/\alpha},\]
which holds by assumption $\delta<2^{1/\alpha}$.

\begin{figure}[ht]
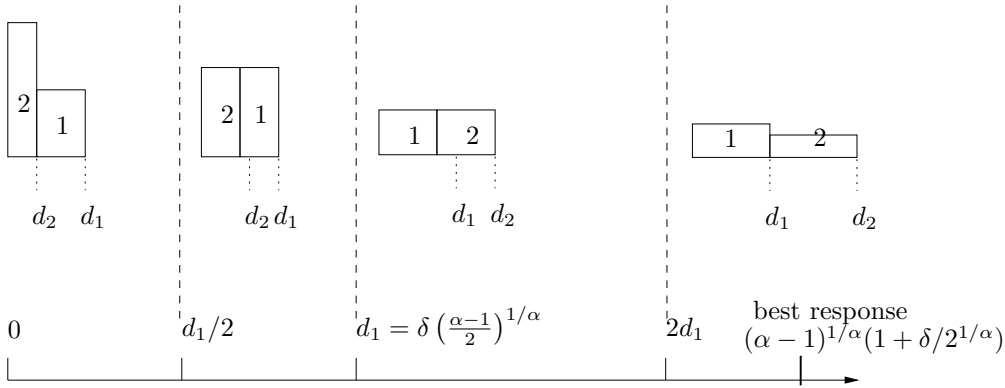
  \caption{Best response for player $2$.}
  \label{fig:convergence}
\end{figure}

In that case the penalty would be
\begin{equation} \label{best2}
d_1 + (\alpha-1)^{\frac{1-\alpha}{\alpha}} + (\alpha-1)^{1/\alpha}
=
(\alpha-1)^{1/\alpha} (1+\delta/2^{1/\alpha} + 1/(\alpha-1)).
\end{equation}
By comparing \eqref{best2} with the remaining case, we show that it is
indeed the best choice for player 2.

\item[Range of $h_3$:] 
The derivative of $h_3$ in $d_2$ is
\[
  1 - 2^{\alpha}(\alpha-1)d_2^{-\alpha},
\]
which is zero for $d_2=2(\alpha - 1)^{1/\alpha}$ and negative for
$d_2 < 2(\alpha - 1)^{1/\alpha}$.  As
\[
  2(\alpha - 1)^{1/\alpha} > 2\delta\left(\frac{\alpha-1}{2}\right)^{1/\alpha} = 2d_1,
\]
the minimum penalty in this range is attained at the right boundary of the interval,
i.e., for $d_2=2d_1 = 2\delta\left(\frac{\alpha-1}{2}\right)^{1/\alpha}$.
But at this point the function $h$, which is continuous, coincides with $h_4$,
which is decreasing in $\left(2d_1,d_1+(\alpha-1)^{1/\alpha}\right)$. Hence
$2d_1$ is not a~local minimum of~$h$.

\item[Range of $h_2$:]
$h_2$ is an increasing function, and therefore it attains a~minimum value only at the lower end of its range, which is $d_1/2$ in this case. Clearly, $d_1/2$ is a~local minimum of~$h$
if and only if $h_1'(d_1/2) \leq 0$ and $h_2'(d_1/2) \geq 0$.
However, we have $h'_2(d_1/2)=h'_1(d_1/2)=1$, so $d_2=d_1/2$ is not a~local minimum of~$h$ either.


\item[Range of $h_1$:]
We will show that $h_1$ is strictly larger than \eqref{best2}.

Since $\delta  < 2^{1/\alpha}$, \eqref{best2} is at most
\[
   (\alpha-1)^{1/\alpha} ( 2+1/(\alpha-1)) = (\alpha-1)^{1/\alpha-1} (2\alpha-1).
\]
To lower bound $h_1$ we use the strict convexity of the function
$x\mapsto x^{1-\alpha}$, which implies
\[
2 \left(\frac{d_2^{1-\alpha}}{2}+\frac{(d_1-d_2)^{1-\alpha}}{2}\right)
>  2 \left( \frac{d_1}{2} + \frac{d_1-d_2}{2} \right)^{1-\alpha}
= 2^{\alpha} d_1^{1-\alpha}. 
\]
Note that $d_1 < (\alpha-1)^{1/\alpha}$ implies
$ d_1^{1-\alpha} > (\alpha-1)^{1/\alpha-1} $.
Combining these, we can finally strictly lower bound the
difference between $h_1$ and the value in \eqref{best2} by
\begin{eqnarray*}
&&d_2+(2^{\alpha}-1) d_1^{1-\alpha} -(\alpha-1)^{1/\alpha-1}(2\alpha-1)\\
  &               >&d_2+(2^{\alpha}-1)
                 (\alpha-1)^{1/\alpha-1}
                 -(\alpha-1)^{1/\alpha-1}(2\alpha-1)\\
&=& d_2 + (\alpha-1)^{1/\alpha-1}
(2^{\alpha}-2\alpha),
\end{eqnarray*}
which is non-negative since $2^{\alpha}\geq 2\alpha$ whenever $\alpha\geq 2$.
\end{description}
\end{proof}

From now on we assume that player 2 chooses $d_2 = d_1+(\alpha-1)^{1/\alpha}=(\alpha-1)^{1/\alpha}(1+2^{-1/\alpha}\delta)$.
What is the best response for player 1?

\begin{lemma} \label{lem:choice1}
Given the second player's choice $d_2$, the penalty of the first player
as a~function of his choice $d_1$ is given by $h(d_1,d_2)$
     and the best response for player 1 is
     \begin{equation*}
     d_1= \delta' \left(\frac{\alpha-1}{2}\right)^{1/\alpha},
     \end{equation*}
     for some $\delta' \in (1,\delta)$.
\end{lemma}

\begin{proof}
  Again player  1 best response  is analyzed through a  case analysis,
  similar   to  the   previous  one.    Figure  \ref{fig:convergence1}
  illustrates  the  best  response  for  player  $1$  when  player $2$  chose
  $d_2=(\alpha-1)^{1/\alpha}(1+2^{-1/\alpha}\delta)$      for     some
  $\delta>1$.   As in  the  previous proof,  for  convenience we  omit
  parameter $d_2$ in each function $h_i$.

\begin{figure}[ht]
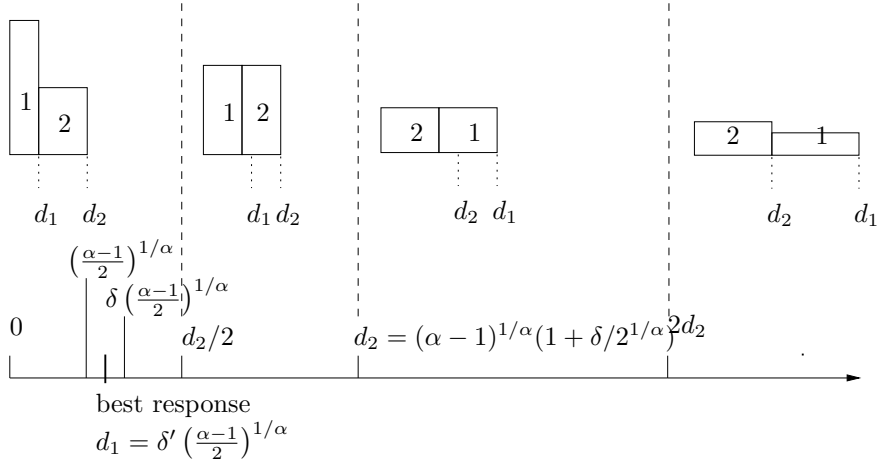
  \caption{Best response for player $1$.}
  \label{fig:convergence1}
\end{figure}

\begin{description}
\item[Range of $h_1$:]
The first derivative of $h_1$ in $d_1$ is
\begin{equation*}
h'_1(d_1) = 1 + (\alpha-1)((d_2-d_1)^{-\alpha} - d_1^{-\alpha})
\end{equation*}
And the second derivative is
\[
h''_1(d_1) = \alpha (\alpha-1) ((d_2-d_1)^{-\alpha-1} + d_1^{-\alpha-1})\]
which is positive, implying that the penalty is convex in $d_1$.

Now, we show that we  have a local minimum for some $1<\delta'<\delta$
at some value
\[
\delta' \left (\frac{\alpha-1}2 \right)^{1/\alpha}.
\]

For this purpose we analyze the interval
\[
\left( \frac{\alpha-1}{2} \right)^{1/\alpha} \leq d_1 \leq
\delta \left( \frac{\alpha-1}{2} \right)^{1/\alpha}.
\]
First   we   evaluate   $h'_1$    at   the   lower   end   $d_1=\left(
  \frac{\alpha-1}{2} \right)^{1/\alpha} $.  In this case
\[
d_2-d_1=  \left(\frac{\alpha-1}{2}  \right)^{1/\alpha}
(2^{1/\alpha} + \delta -1).
\]
This means that
\begin{eqnarray*}
h'_1(d_1)& = &1 +  (\alpha-1)\frac{2}{\alpha-1} (
2^{1/\alpha} + \delta - 1)^{-\alpha}
- (\alpha-1)  \frac{2}{\alpha-1} \\
&=& 2 (2^{1/\alpha} + \delta - 1)^{-\alpha} -1 < 0 \enspace,
\end{eqnarray*}
as $\delta>1$ and $\alpha>1$.

Secondly, we evaluate the $h'_1$ at
the upper end
\begin{equation}
\label{upper1}
d_1= \delta \left( \frac{\alpha-1}{2} \right)^{1/\alpha},
\end{equation}
then by $d_2-d_1 = (\alpha-1)^{1/\alpha}$
we obtain
\[
1 + (\alpha-1) \frac{1}{\alpha-1} - (\alpha-1)\frac2{\alpha-1} \delta^{-\alpha}
= 2 - 2 /\delta^{\alpha} > 0 \enspace,
\]
as $\delta  <  2^{1/\alpha}$  and  $\alpha  >1$.
Together with the continuity of the penalty function, it implies
that there is a value $1<\delta'<\delta$ such that
\[d_1= \delta' \left(\frac{\alpha-1}{2}\right)^{1/\alpha} \]
is a local minimum.

To conclude we compare this local minimum with the remaining cases, showing the dominance of this value.

\item[Range of $h_2$:] $h_2$ is an increasing function and therefore is minimum  at $d_1=d_2/2$. Again, $d_1/2$ is a~local minimum of~$h$ if and only if $h_1'(d_2/2) \leq 0$ and $h_2'(d_2/2) \geq 0$. However, we have $h'_2(d_2/2)=h'_1(d_2/2)=1$, so $d_1=d_2/2$ is not a~local minimum of~$h$ either.




\item[Range of $h_3$:]
In this case, the derivative of the penalty function is
\[1-(\alpha-1)2^{\alpha}d_1^{-\alpha}\]
which is zero at
\begin{equation}
  \label{eq:zero}
d_1=2(\alpha-1)^{1 / \alpha}.
\end{equation}
Note that the second derivative
\[
 h''_3(d_1) = \alpha(\alpha-1)2^{\alpha}d_1^{-\alpha-1} > 0
\]
for $\alpha>1$, so the penalty is convex, and \eqref{eq:zero} is a local minimum.
It is greater than $d_2$ by $\delta<2^{1/\alpha}$ and smaller than
$2d_2$ by $\delta>0$. Therefore the local minimum belongs to
  the range considered in this case.

The penalty at $d_1=2(\alpha-1)^{1 / \alpha}$ is
\begin{eqnarray*}
& &        2(\alpha-1)^{1/\alpha}        +        2^\alpha        \cdot
2^{1-\alpha}(\alpha-1)^{1/\alpha-1}      -     (\alpha-1)^{1/\alpha-1}
(1+\delta/2^{1/\alpha})^{1-\alpha}
\\
&=&            (\alpha-1)^{1/\alpha-1}            \left(2\alpha           -
(1-\delta/2^{1/\alpha})^{1-\alpha} \right).
\end{eqnarray*}

We  claim that  this penalty  is  larger  than $h_1$ evaluated at \eqref{upper1}, eliminating $2(\alpha-1)^{1 / \alpha}$ for a best response.  For this purpose we need to
show
\begin{eqnarray*}
h_1\left(\delta \left( \frac{\alpha-1}{2} \right)^{1/\alpha}\right)&<&h_3(2(\alpha-1)^{1 / \alpha})\\
\Leftrightarrow \frac{\delta(\alpha-1)}{2^{1/\alpha}}                                 +
\frac{2^{1-1/\alpha}}{\delta^{\alpha-1}} + 1 &< & 2\alpha\\
\Leftrightarrow\delta^\alpha     (\alpha-1)     +    2     &<&(2\alpha-1)2^{1/\alpha}
\delta^{\alpha-1}.
\end{eqnarray*}
This holds since
\begin{align*}
 \delta^{\alpha-1}(2\alpha-1)2^{1/\alpha} - \delta^\alpha (\alpha-1)
   = \delta^{\alpha-1} \left( 2^{1/\alpha}(2\alpha-1) - \delta (\alpha-1) \right) \\
   > 2^{1/\alpha}(2\alpha-1) - \delta (\alpha-1) > 2^{1/\alpha}\alpha > 2,
\end{align*}
for any $\alpha>1$ because $2^{1/\alpha}\alpha$ increases with $\alpha>\ln 2$ and equals $2$ when $\alpha=1$.

\item[Range of $h_4$:]
The derivative of $h_4$ in $d_2$ is
\[
 1-(\alpha-1)(d_1-d_2)^{-\alpha},
\]
which is zero exactly for the value
\[
d_2 + (\alpha-1)^{1/\alpha}.
\]
As the second derivative $h''_1(d_1) = \alpha(\alpha-1)(d_1-d_2)^{-\alpha-1}$ is
positive, the choice $d_1=d_2 + (\alpha-1)^{1/\alpha}$ minimizes the penalty among $d_2\geq
2d_1$. However, $h_1$ has a~local minimum if
\[(\alpha-1)^{1/\alpha} \geq d_2=(\alpha-1)^{1/\alpha}(1+2^{-1/\alpha}\delta),\]
which is a contradiction by $\delta2^{-1/\alpha}>0$
\end{description}
\end{proof}

This concludes the proof of Theorem~\ref{thm:convergenceUnbounded}.
\end{proof}

\subsection{Bounding total charge}

In this section we bound the total cost share for the \textsc{Marginal
  Cost   Sharing   Scheme},  by   showing   that   it   is  at   least
$E(\mathrm{OPT}(d)) $ and at most  $\alpha$ times this value.  In fact
we show a stronger claim for individual cost shares.

\begin{theorem}  \label{thm:overCharge}  For  every  player  $i$,  its
  marginal  costshare is at  least its  proportional costshare  and at
  most $\alpha$ times the proportional costshare.
\end{theorem}
\begin{proof}
  Fix a player $i$, and denote  by $S_{-i}$ the schedule obtained
  from  $\mathrm{OPT}(d)$ when all  executions of  $i$ are  replaced by  idle times.
  Clearly we have the following inequalities.
\[
   E(\mathrm{OPT}(d_{-i})) \leq E(S_{-i}) \leq E(\mathrm{OPT}(d))
\]
Then the marginal cost share of player $i$ can be lower bounded by
\[
   E(\mathrm{OPT}(d))-E(\mathrm{OPT}(d_{-i})) \geq E(\mathrm{OPT}(d))-E(S_{-i}).
\]

According to~\cite{YaoDemmersShenker1995}  the schedule $\mathrm{OPT}$
can be obtained by the  following iterative procedure.  Let $P$ be the
support of a  partial schedule. For every interval  $[t,t')$ we define
its domain $I_{t,t'} := [t,t') \backslash P$, the set of included jobs
$J_{t,t'} :=  \{j :  [r_j,d_j) \subseteq [t,t')  \}$, and  the density
$\sigma_{t,t'}  := \sum_{j \in  J_{t,t'}} w_j  / |  I_{t,t'} |  $.  The
procedure  starts  with $P=\emptyset$,  and  while  not  all jobs  are
scheduled,  selects an  interval  $[t,t')$ with  maximal density,  and
schedules  all jobs  from  $J_{t,t'}$ in  earliest  deadline order  in
$I_{t,t'}$ at speed $\sigma_{t,t'}$, then  adds $I_{t,t'}$ to $P$.

For the upper bound, let $t_1<t_2<\ldots<t_\ell$ be the sequence of
all release times and deadlines for some $\ell\leq 2n$. For convenience, we denote $S=\mathrm{OPT}(d)$ and $S'=\mathrm{OPT}(d_{-i})$.  Clearly both
schedules   run   at   uniform   speed    in   every   interval
$[t_{k-1},t_k)$. For every  $1\leq k\leq n$ let $s_k$  be the speed of
$S$ in  $[t_{k-1},t_{k})$, and  $s'_k$ the speed  of $S'$ in  the same
interval.

From the  iterative procedure described above  it follows that  every job is  scheduled at
constant speed. Let  $s_0$ be the speed at which  job $i$ is scheduled
in  $S$. It  also follows  that if
$s_k>s_0$, then $s'_k = s_k$,  and if $s_k\leq s_0$, then $s'_k\leq s_k$.

We establish the following upper bound.
  \begin{align*}
    E(\mathrm{OPT}(d))-E(\mathrm{OPT}(d_{-i}))  &  =   \sum_{k=1}^{\ell}  s_k^{\alpha}  (t_k-t_{k-1})  -
    s'^{\alpha}_k (t_k-t_{k-1})
\\
    & = \sum_{k=1}^{\ell}  (t_k-t_{k-1}) (s_k^\alpha - (s_k -(s_k-s'_k))^\alpha )
\\
    &  = \sum_{k=1}^{\ell}  (t_k-t_{k-1}) s_k^\alpha  \left(1 - \left  (1
        -\frac{s_k-s'_k}{s_k} \right)^\alpha \right)
\\
    & \leq \sum_{k=1}^{\ell} (t_k-t_{k-1}) s_k^\alpha \left (1 -  \left (1 -\alpha \frac{s_k-s'_k}{s_k}\right) \right)
\\
    & = \sum_{k=1}^{\ell} (t_k-t_{k-1}) \alpha s_k^{\alpha-1} (s_k-s'_k)
\\
    & \leq \alpha s_a^{\alpha-1}\sum_{k=1}^{\ell} (t_k-t_{k-1}) (s_k-s'_k)
\\
    & = \alpha s_a^{\alpha-1} w_i
\\
& = \alpha (E(\mathrm{OPT}(d))-E(S_{-i})).
  \end{align*}

The first  inequality uses  the generalized Bernoulli  inequality, and
the last  one the fact  that for all  $k$ with $s_k\neq s'_k$  we have
$s_k\leq s_a$.

The  theorem  follows  from   the  fact  that  $s_a^{\alpha-1}w_i$  is
precisely the proportional cost share of job $i$ in $\mathrm{OPT}(d)$.
\end{proof}


A  tight example  is  given by  $n$  jobs, each  with workload  $1/n$,
release time $0$ and
deadline $1$. Clearly  the optimal energy consumption is  $1$ for this
instance.   The    marginal   cost   share   for    each   player   is
$1-(1-1/n)^\alpha$. Finally we observe that the total marginal cost share
tends to $\alpha$, i.e.\
\[
  \lim_{n\rightarrow +\infty} n-n(1-1/n)^\alpha = \alpha.
\]


\section{Final remarks}

\subsection{Cross-monotonicity}
The \emph{cross-monotonicity} is a property of  cost sharing games, stating that whenever new players enter the game, the cost share of any
fixed player does not increase.  This property is useful for stability
in  the  game,  and  is  the  key  to  the  Moulin  carving  algorithm
\cite{MoulinShenker:01:Strategyproof-sharing}, which  selects a set of
players to be served for specific games.

In the game  that we consider, the minimum  energy of an optimal schedule
for a set $S$ of jobs contrasts  with many studied games,
where serving  more players becomes  more cost effective,  because the
considered equipment is better used.

Consider a  very simple example  of two identical  players, submitting
their respective jobs with the same deadline $1$. Suppose the workload
of each job is $w$, then  the minimum energy necessary to schedule one
job is $w^\alpha$, while the cost to serve both jobs is $(2w)^\alpha$,
meaning that the  cost share increase whenever a  second player enters
the game. Therefore the marginal cost sharing scheme is not cross-monotonic.


\subsection{Uniqueness of Nash equilibria}

In this paper we showed that the deadline game with the marginal cost sharing mechanism always admits a pure Nash equilibrium.  However the Nash equilibrium may not be unique. Here, a simple example is an instance with $n$ identical players where $n!$ Nash symmetric equilibria are admitted. For arbitrary instances, the uniqueness of Nash Equilibrium raises two questions.  The first question concerns the comparison of different Nash equilibria in terms of social cost.  If the divergence is significant, then it means that the outcome of the game can be arbitrary, and may indicate the need for another mechanism, which smooths the possible resulting Nash equilibria.  The second question concerns the characterization of job sets which lead to a unique Nash equilibrium.

We are interested in this last question, already in the 2 player setting.  Here we fixed normalized quantities $p_1=1,w_1=1$ for the first player, and leave the quantities of the second player variable.  Which points $(p_2,w_2)$ do admit a unique Nash equilibrium?

The 2 potential Nash equilibria are the following strategy profiles
\begin{description}
  \item[S21] $d_2=\ell_2^*, d_1=\ell_2^*+\ell_1^*$
  \item[S12] $d_1=\ell_1, d_2=\ell_1+\ell_2$,
\end{description}
where the lengths $\ell_1,\ell_2,\ell_1^*,\ell_2^*$ are defined as follows.
\begin{align}
\ell_2^* &= w_2 \frac{(\alpha - 1)^{1/\alpha}}{(p_1 + p_2)^{1/\alpha}}
&
\ell_1^* &= w_1 \frac{(\alpha - 1)^{1/\alpha}}{(p_1)^{1/\alpha}}
\label{eq:S21} \\
\ell_1 &= w_1 \frac{(\alpha - 1)^{1/\alpha}}{(p_1 + p_2)^{1/\alpha}}
&
\ell_2 &= w_2 \frac{(\alpha - 1)^{1/\alpha}}{(p_2)^{1/\alpha}}
\label{eq:S12}.
\end{align}

To break the symmetry we consider only points where the social cost of S21 is minimal among the two profiles. In \cite{durr2014penaltyscheduling} we showed that these are precisely the points that satisfy
\[
  w_2 \leq  w_1 \frac{ (p_1+p_2)^{(\alpha-1)/\alpha} - p_1^{(\alpha-1)/\alpha} }
                      { (p_1+p_2)^{(\alpha-1)/\alpha} - p_2^{(\alpha-1)/\alpha} }.
\]
For such points, we ask whether one of the players wants to deviate from S12, i.e., the other potential equilibrium, with larger social cost.
Our experiments indicate that, for each player $j\in\{1,2\}$,
there is a threshold $t_j$, which depends on $p_2$, such that player $j$
wants to deviate if and only if $w_2 \leq t_j$.
Figure~\ref{fig:2playerUnique} depicts the plots of these threshold functions,
determined numerically.  We were unable to rigorously prove their existence,
obtain their closed forms, or relate them to functions studied
in~\cite{durr2014order}.  However by choosing $\alpha=2$, $p_1=w_1=1$ and
$(p_2,w_2)$ inside the shaded region of
Figure~\ref{fig:2playerUnique}, one can easily verify that S21 is a unique pure Nash equilibrium.

\begin{figure}[t]
	\centerline{\includegraphics[height=8cm]{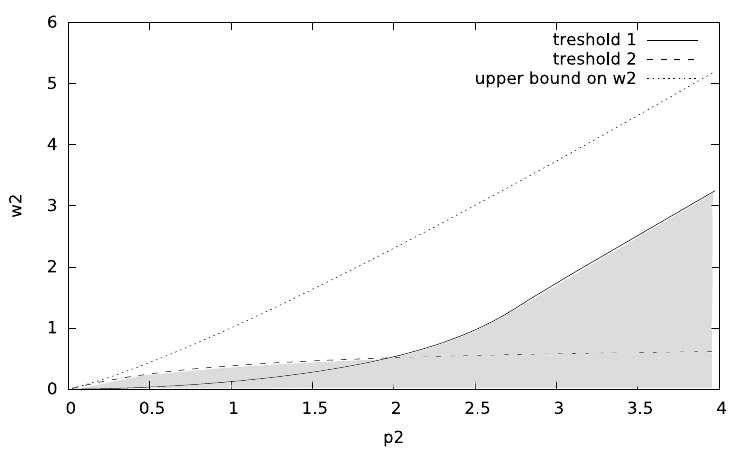}}
	\caption{Experimental results for $\alpha=2$. Horizontal axis is $p_2$, and vertical axis is $w_2$.  Depicted are the upper bound for $w_2$ and the thresholds $t_1,t_2$.}
	\label{fig:2playerUnique}
\end{figure}

\section*{Acknowledgements}
We would like to thank anonymous referees for remarks and suggestions on an earlier version of this paper.

\medskip\noindent
Christoph Dürr and Oscar C. Vásquez were partially supported by grant ANR-11-BS02-0015. Oscar C. Vásquez was partially supported by FONDECYT grant 11140566. {\L}ukasz Je{\.z} was partially supported ERC consolidator grant 617951, Israeli Centers of Research Excellence (I-CORE) program, Center No.4/11, NCN grant DEC-2013/09/B/ST6/01538, and FNP Start scholarship.

\bibliographystyle{plain}
\bibliography{MarginalTCS}
\end{document}